\documentclass[11pt]{article}

\usepackage[utf8]{inputenc}
\usepackage[T1]{fontenc}
\usepackage{graphicx,amsthm}
\usepackage{amsmath}
\usepackage{fancyhdr}
\usepackage{amsfonts}
\usepackage{amssymb}
\usepackage{enumerate}
\usepackage{tikz}
\usepackage{subfigure}

\usepackage{xcolor}
\usepackage{graphicx}
\usepackage{tabularx}
\usepackage[center]{caption}
\usepackage{array}

\usepackage{multirow}

        \theoremstyle{plain}

\def\ds{\displaystyle}
\newcommand{\R}{{\mathbb{R}}}

\newcommand{\divv}{\nabla_v \cdot}
\newcommand{\divx}{\nabla_x \cdot}

\newcommand{\dt}{\partial_t}

\newcommand{\p}{\partial}

\newcommand{\Kn}{\mathrm{Kn}\thinspace}
\newcommand{\la}{\left\langle}
\newcommand{\ra}{\right\rangle}

\newcommand{\cint}[1]{\langle #1 \rangle}

\def\eps{\varepsilon}
\usepackage{amsmath}
\usepackage{amsthm}

        \theoremstyle{plain}
        
        \newtheorem{proposition}{Proposition}[section]

        \theoremstyle{remark}
        
        \theoremstyle{remark}
        
        \theoremstyle{remark}

\begin{document}

\begin{center}
{\bf BGK and Fokker Planck Models   for thermally perfect gases}

\vspace{1cm}
J. Mathiaud$^{1,2}$,  L. Mieussens$^1$

\bigskip
$^1$Univ. Bordeaux, Bordeaux INP, CNRS, IMB, UMR 5251, F-33400 Talence, France.\\
{ \tt(Luc.Mieussens@math.u-bordeaux.fr)}

\bigskip
$^2$Centre  Lasers  Intenses  et  Applications, Université  de  Bordeaux-CNRS-CEA,UMR  5107, , F-33400 Talence, France. { \tt(julien.mathiaud@u-bordeaux.fr)}\\

\end{center}

\begin{abstract}
We propose two models of the Boltzmann equation (BGK and Fokker-Planck
models) for rarefied flows of thermally perfect gases. These
models take into account various models of energy, which are required for high temperature flows, like for
atmospheric re-entry problems as long as the pressure law for perfect gases is true. We prove that these models satisfy
conservation and entropy properties (H-theorem), and we derive their
corresponding compressible Navier-Stokes asymptotic.
\end{abstract}

 \bigskip

Keywords: Fokker-Planck model, BGK model, H-theorem, Rarefied Gas
Dynamics, thermally perfect gases

\section{Introduction}

Numerical simulation of atmospheric reentry flows requires to solve
the Boltzmann equation of Rarefied Gas Dynamics. The standard method
to do so is the Direct Simulation Monte Carlo (DSMC) method
\cite{bird,BS_2017}, which is a particle stochastic method. However,
it is sometimes interesting to have alternative numerical methods,
like, for instance, methods based on a direct discretization of the
Boltzmann equation (deterministic approaches). This is hardly possible
for the full Boltzmann equation (except for monatomic gases,
see~\cite{luc_2014,herouardgal}), since this is still much too computationally
expensive for real gases. But BGK like model equations~\cite{bgk} are
very well suited for such deterministic codes: indeed, their
complexity can be reduced by the well known reduced distribution
technique~\cite{Chu_1965}, which leads to intermediate models between
the full Boltzmann equation and moment
models~\cite{Struchtrup_moment_book}. The Fokker-Planck
model~\cite{cercignani} is another model Boltzmann equation that can
give very efficient stochastic particle methods, see~\cite{Grj2011}.

These model equations have already been extended to polyatomic
gases, so that they can take 
into account the internal energy of rotation of gas molecules. They
contains correction terms that lead to correct transport coefficients: the ESBGK or
Chekhov's models~\cite{Holway,esbgk_poly,S_model}, and the cubic
Fokker-Planck and
ES-Fokker-Planck~\cite{Grj2011, Grj2013,Mathiaud2016,Mathiaud2017}.

For high temperature flows, like in space reentry problems, other energies can be  activated (like vibrations: \cite{hdr,Mathiaudvib}) and have a significant
influence on energy transfers in the gas flow. It is therefore
interesting to extend the model equations to take these energies into account. Several extended BGK models have been
recently proposed to do so, for
instance~\cite{RS_2014,WYLX_2017,ARS_2017,KKA_2019}, and a 
Fokker-Planck model was proposed earlier in~\cite{Grj2013}.

In this paper, we create BGK and Fokker-Planck models for every thermally perfect gas (perfect gas with energy depending on temperature without knowing the precise form of the dependence).  Since no obvious kinetic description of degrees of liberty can be precisely given for such gases, we directly use a reduced model with one function for translation energy and the other one for other degrees of freedom: note that with this reduction, only higher order moments with respect to the vibration energy variable are lost: the macroscopic quantities of interest like pressure, temperature, and heat flux,are the same as in the non-reduced model. Moreover, since the reduced variable is not the velocity,this reduction does not require any assumption or special geometries. We prove that these reduced models satisfies the
H-theorem as well as conservation properties. This paper is a first step towards the creation of ESBGK or ES-Fokker-Planck models able to reproduce the 
relaxation of several energies towards equilibrium

Our paper is organised as follows. In section~\ref{sec:kinetic}, we
present the kinetic description of a thermally  perfect gas and we discuss the
mathematical properties of the reduced distributions that will be used for our
models. Our BGK and Fokker-Planck models are presented in
sections~\ref{sec:BGK} and~\ref{sec:FP}, respectively. In
section~\ref{sec:CE}, the hydrodynamic limits of our models, obtained by a
Chapman-Enskog procedure, are discussed. In section~\ref{sec:extens} we provide an extension of our framework to several energies. Finally section~\ref{sec:conclusion}
gives some perspectives around this work.

\section{Kinetic description of a thermally perfect gas}
\label{sec:kinetic}
\subsection{Some thermodynamics on thermally perfect gases}

Before writing any kinetic model, we consider some  thermodynamics.
A thermally perfect gas is a gas satisfying $P=\rho RT$ where $P$ is the pressure, $\rho$ the density and $T$ the temperature of a gas with $R$ , constant of the perfect gas. However this law does not give the relation between energy and temperature which can be linear (for simple models of rotations for instance) or non linear (for instance for vibrations in a diatomic molecule can be set to $\ds e(T)=\frac52RT+\frac {RT_0}{\exp(T_0/T)-1}$ for some $T_0$ characteristic temperature of vibrations) or simply tabulated according to temperature. In simple cases when the relation between energy  and temperature is well defined one can construct BGK or Fokker-Planck models to capture correctly the physics. We do want to extend these models to any kind of energy. In order to do so we suppose that energy $e$ can be defined through $e(T)=e_{tr}(T)+e_{int}(T)$ where:\\
\begin{equation}
    e_{tr}(T)=\frac32RT \label{eq-etr},e_{int}(T)=e(T)-e_{tr} (T) ,
\end{equation} $ e_{tr}$ being the translational kinetic energy  and $e_{int}$
 represents all the other internal energies . Moreover we also suppose that  $e_{int}$ is a strictly increasing function of temperature: $e$ is then also a strictly increasing function so that there exists   one-to-one  functions $\mathbb T$  and $\mathbb T_{int}$ such that
\begin{eqnarray}
 T=\mathbb T(e) \qquad,&\qquad e=e(T)\qquad &, \qquad \label{eq-defT}  \frac{d}{dT} e=\frac32R+c_v^{int}>0,\\
  T_{int}=\mathbb T_{int}(e_{int}) \qquad,&\qquad e_{int}=e_{int}(T_{int})\qquad \label{eq-defTint}&,\qquad\frac{d}{dT_{int}}e_{int}=c_v^{int}>0,
\end{eqnarray}
where $c_v^{int}$ is the specific heat associated to $e_{int}$. 

\

We can also define an entropy $s^{int}$ satisfying $\ds ds^{int}=\frac{de_{int}}{\mathbb T_{int}(e_{int})}$ up to some constant by integration (we will give some expressions later for the simple rotational and vibrational case). Similarly an entropy  for translations can be defined through $\ds ds^{tr}=\frac{de_{tr}}{\mathbb T_{tr}(e_{tr})}$.
The second principle   now writes:
$$ds(\rho,T)=ds^{tr}(T)+ds^{int}(T)-R\frac{d\rho}{\rho}.$$

We now have now all the necessary tools to construct our BGK and Fokker-Planck models.
 
\subsection{Distribution function and local equilibrium}
We consider a thermally perfect gas. Since we only know the relation between  temperature and  energy, there is no clear extension to other degrees of freedom than the translational ones for an equilibrium state such as the one of polyatomic gases (\cite{esbgk_poly}) or vibrational diatomic gases (\cite{mathiaud2019BGK}). We propose to separate the translational degrees of freedom of molecules with the remaining degrees of freedom: even if there is a loss of information for high order moments in internal energy this reduction will be enough to capture both pressure and thermal flux which are the quantities of interest in our problem. To do that we define $F(t,x,v)$ the mass
density of
molecules with position $x$, velocity $v$ and $G(t,x,v)$ the internal energy density of molecules with position $x$, velocity $v$. We directly write a two model distribution as it was done in (\cite{esbgk_poly,mathiaud2019BGK}).
The corresponding local equilibrium
distributions for $F$ and $G$  are defined by (see \cite{bird})
\begin{eqnarray}
 M_{int}[F,G](v) &=&\frac{\rho}{\sqrt{2\pi RT}^3}  \exp\left(-\frac{\frac12|u - v|^2}{ RT}\right) ,\\ 
e_{int}(T) M_{int}[F,G](v) &=&e_{int}(T)\frac{\rho}{\sqrt{2\pi RT}^3}  \exp\left(-\frac{\frac12|u - v|^2}{ RT}\right) .
 \label{maxvib}
\end{eqnarray}

Here, $\rho$ is the gas density, $T$ its equilibrium temperature and $u$ its mean velocity, defined through:
\begin{eqnarray}  \label{eq-mtsred} 
 &&\rho = \la F \ra_{v}\qquad ,\qquad
 \rho u = \la v F \ra_{v},\qquad\\
&&  \rho e_{tr} = \la (\frac12(v-u)^2)  F\ra_{v},  \rho e_{int} = \la G\ra_{v},\\
&& \rho e = \la (\frac12(v-u)^2)  F\ra_{v}+\la G\ra_{v},\ T=\mathbb T(e)
\end{eqnarray}
where we use the notation $ \la \psi
\ra_{v}=\iint \psi(t,x,v)\, dv $ for any function $\psi$.

\

Immediate computations of Gaussian functions lead to the following proposition:

\begin{proposition}[Conservation properties]
  \begin{eqnarray*}   
 &&\rho =  \la  M_{int}[F,G]\ra_{v},\qquad
 \rho u = \la v  M_{int}[F,G] \ra_{v},\qquad\\
&&  \rho e_{tr} = \la (\frac12(v-u)^2)   M_{int}[F,G]\ra_{v},  \rho e_{int} = \la  e_{int}(T)M_{int}[F,G]\ra_{v},\\
&& \rho e = \la (\frac12(v-u)^2)   M_{int}[F,G]\ra_{v}+\la e_{int}(T)M_{int}[F,G]\ra_{v},\ T=\mathbb T(e).
\end{eqnarray*}
\end{proposition}

\

We now define a reduced entropy as a function of $F$
and $G$   in the following proposition:
\begin{proposition}[Entropy]\label{entrop}
We define the following reduced entropy $\mathcal{H}(F,G)$ of $F$ and $G$ as: 
\begin{eqnarray}
\mathcal{H}(F,G)= \la F\log(F)- F\frac{s_{int}}R\left(\frac G F\right) \ra_{v}. \label{eq-HFG}
\end{eqnarray}
\begin{enumerate}
\item  The partial derivatives of $H=F\log(F)- F\frac{s_{int}}R\left(\frac G F\right)$ computed at $(F,G)$ are:
\begin{equation}  \label{eq-partialH}
  D_1H(F,G)=1+\log(F)+\frac{G}{R{\mathbb T}_{int}(G/F)F}-\frac{s_{int}}R\left(\frac G F\right),
\quad D_2H(F,G)=-\frac1{R{\mathbb T}_{int}(G/F)}.
\end{equation}

\item We note $\ds \mathbb{H} =
  \left(\begin{smallmatrix}
D_{11}H(F,G)    & \quad D_{12}H(F,G) \\
 D_{12}H(F,G)  & \quad D_{22}H(F,G)
  \end{smallmatrix}\right)$
 the Hessian matrix of $H$.  Its value is:
$$ \begin{pmatrix}
 D_{11}H(F,G)=\frac 1 F +\frac{G^2}{F^3{c_v}_{int}R{\mathbb T}^2_{int}(G/F)}, & D_{12}H(F,G)=-\frac {G}{F^2{c_v}_{int}R{\mathbb T}^2_{int}(G/F)}\\
 D_{21}H(F,G)=D_{12}H(F,G), & D_{22}H(F,G)=\frac{1}{{c_v}_{int}R{\mathbb T}^2_{int}(G/F)F}
\end{pmatrix} $$
 Moreover, the derivatives satisfy the following equality's:
\begin{equation}\label{eq-FDG} 
\begin{split}
& FD_{11}H(F,G)+GD_{21}H(F,G) = 1,\\
& FD_{12}H(F,G)+GD_{22}H(F,G) = 0.
\end{split}
\end{equation}

\item The function $(F,G) \mapsto H(F,G)$ is convex.
\end{enumerate}
\end{proposition}
\begin{proof}
Points 1 and 2 are given by direct computations. The Hessian matrix is positive
definite (trace and determinant are positive) so that  $H$ is convex.
\end{proof}
\begin{proposition}[Minimisation of entropy]  \label{minentrop} Let $(F,G)$ be a couple of reduced distributions and $\rho$, $\rho u$, and $\rho e$ its moments as defined by~\eqref{eq-mtsred}. Let ${\cal S}$ be the convex set defined by
\begin{equation*}
    \mathcal{S}=\left\{(F_1,G_1) \text{ such that} \la F_1\ra_{v}=\rho ,\quad
  \la vF_1\ra_{v}=\rho u, \quad 
  \la \frac12|v|^2 F_1+G_1\ra_{v} = \rho e \right\}.
\end{equation*}
\begin{enumerate}
\item The minimum of 
  $\mathcal{H}$ on ${\cal S}$
is obtained for the couple $(M_{int}(F,G), e_{int}(T)M_{int}(F,G))$
with
\begin{equation}  \label{eq-MvibFG}
M_{int}[F,G]=\frac{\rho}{\sqrt{2\pi RT}^3} \exp\left(-\frac{|v - u|^2}{2 RT}\right)  
 \end{equation}
where $e_{int}(T)$ is the equilibrium internal energy defined by ~(\ref{eq-etr}).

\item For every $(F_1,G_1)$ in $\mathcal{S}$, we have
\begin{equation*}
\begin{split}
 0 & \ge  \ds H(M_{int}(F,G), e_{int}(T)M_{int}(F,G)) -H(F_1,G_1) \\
&\ge  D_1H(F,G)(M_{int}(F,G)-F)+D_2H(F,G)(e_{int}(T)M_{int}(F,G)-G)
\end{split}
\end{equation*}
\end{enumerate}

\end{proposition}

\begin{proof} 

 We now compute the minimum of the reduced entropy. First, the set $\mathcal{S}$ is
clearly convex, and it is non empty, since it is easy to see that
$(M_{int},e_{int}(T)M_{int})$ realises the moments $\rho$, $\rho u$,
and $\rho e$, and hence belongs to $\mathcal{S}$. 
Now, we define the following Lagrangian
\begin{equation*}
\begin{split}
\mathcal {L}(F_1,G_1,\alpha,\beta,\gamma)=& \la
  H(F_1,G_1)\ra_{v}-\alpha(\la F_1\ra_{v}-\rho) \notag\\
 & -\beta\cdot(\la vF_1\ra_{v}-\rho u) -\gamma\left(\la
(\frac12|v|^2)F_1+G_1\ra_{v} - \rho e\right)
\end{split}
\end{equation*}
for $(F_1,G_1)\in\mathcal{S}$ , $\alpha\in \mathbb{R}$, $\beta\in
\mathbb{R}^3$, $\gamma\in \mathbb{R}$. The reduced entropy can reach   a
minimum of ${\mathcal S}$ when $\mathcal {L}$ has its first derivatives equal to zero. This  point, denoted by $(F_1,G_1,\alpha,\beta,\gamma)$
for the moment, is characterised by the fact that the partial
derivatives of $\mathcal {L}$ vanish at
$(F_1,G_1,\alpha,\beta,\gamma)$. This gives the following relations:
\begin{eqnarray}
&&D_1H(F_1,G_1)=\alpha+\beta\cdot v+\gamma \frac12|v|^2, \label{Ent1}\\
&&D_2H(F_1,G_1)=\gamma,  \label{Ent2}\\
&&\la F_1\ra_{v}-\rho=0,\label{Ent3} \\
&&\la v F_1 \ra_{v}-\rho u=0, \label{Ent4}\\
&&\la (\frac12|v|^2)F_1+G_1 \ra_{v}-\rho e=0\label{Ent5},
\end{eqnarray}
where $D_1H$ and $D_2H$ are defined in~(\ref{eq-partialH}).
Combining equations (\ref{Ent1}) and (\ref{Ent2}), one gets that there
exist real numbers $A$, $B$, $D$ and one vector $E\in
\mathbb{R}^3$, independent of $v$, such that:
\begin{eqnarray*}
F_1&=&A\exp\left(E\cdot v +B|v|^2\right ),\\
G_1&=&D F_1,
\end{eqnarray*}
where $B$  is necessarily non positive to ensure the
integrability of $F_1$ and $G_1$. $G/F$ is a constant because the temperature is 
a bijective function of energy so that $D_2H(F_1,G_1)=\gamma$  only owns one solution.
It is then standard to use equations~(\ref{Ent3}) to get
$F_1=M_{int}(F,G)$ and $G_1=e_{int}(T)M_{int}(F,G)$.

Finally point 2 is a direct consequence of the convexity of $H$
and of the minimization property.
\end{proof}

\section{A BGK model for thermally perfect gases}
\label{sec:BGK}
\subsection{A reduced BGK model}
For physics considerations, it is interesting to reduce  complex
kinetic models by using the usual reduced distribution
technique~\cite{HH_1968}. Even for perfect gases, some degenerate models of energy cannot be described easily through equilibrium (extensions of Maxwellians is not very clear when ones deals with partial degrees of freedom of internal energy. In this paper we propose to use two reduced functions $F$ and $G$ that will transport energy. $F$ is transporting the translational energy whereas $G$ transport the reminder of internal energy. More precisely they are defined through:
\begin{equation}
\begin{aligned}
& \p_t F + v\cdot\nabla_x F  =  \frac1 \tau \left( M_{int}[F,G] - F \right) \, , \label{eq:eq_marginales_fint}\\
&  \p_t G + v\cdot\nabla_x G  = \frac1 \tau ( e_{int} M_{int}[F,G] - G ) \, ,   
\end{aligned}
\end{equation}
where the reduced Maxwellian is
\begin{equation*}
M_{int}[F,G]=\frac{\rho}{\sqrt{2\pi RT}^3} \exp\left(-\frac{|v - u|^2}{2 RT}\right),
\end{equation*}

and the macroscopic quantities are defined by

\begin{equation}\label{eq-macro_reduced} 
 \rho = \la F \ra_{v}, \qquad \rho u= \la vF \ra_{v},
 \qquad 
 \rho e= \la (\frac12(v-u)^2)F \ra_{v}+\la G\ra_{v},
\end{equation}

and $T$ is still defined by~\eqref{eq-defT}.

\

It is interesting to compare our new model to the work
of~\cite{hdr,mathiaud2019BGK} and~\cite{KKA_2019}: in these recent papers, the
authors also proposed, independently, BGK and ES-BGK models for temperature
dependent $\delta$, like in the case of vibrational energy. However,
they are not based on an underlying discrete vibrational energy
partition, and the authors are not able to prove any H-theorem. Only
a local entropy dissipation can be proved. The advantage of our 
approach is that the reduced model, which is continuous in energy too,
has got a
H-theorem, as it is shown below.

\subsection{Properties of the reduced model}

System~(\ref{eq:eq_marginales_fint})
naturally satisfies local conservation laws of mass, momentum, and
energy. Moreover, the H-theorem holds with the reduced entropy 
$H(F,G)$ as defined in~(\ref{eq-HFG}). Indeed, we have the 
\begin{proposition}
The reduced BGK
system~(\ref{eq:eq_marginales_fint})
satisfies the H-theorem
\begin{equation*}
\partial_t {\cal H}(F,G) +\divx\la v H(F,G)\ra_{v}\leq 0,
\end{equation*}
where  ${\cal H}(F,G)$ is the reduced entropy defined in~\eqref{eq-HFG}.
\end{proposition}

\begin{proof}

\

By differentiation we get
\begin{equation*}
\begin{split}
& \partial_t{\cal H}(F,G) +\divx\la v H(F,G)\ra_{v} \\
&  =  
\la  D_1H(F,G)(\partial_t F + v\nabla_x F)  
  +  D_2H(F,G)(\partial_t G + v\nabla_x G)
\ra_{v}  \\
&  =  
 \frac1 \tau \la  D_1H(F,G) ( M_{int}[F,G] - F )
  +  D_2H(F,G)(\frac {\delta(T)} 2 RT M_{int}[F,G] - G )
\ra_{v}  \\
&  \leq  0
\end{split}
\end{equation*}
where we have
used~(\ref{eq:eq_marginales_fint}) to
replace the transport terms by relaxation ones, and point 2 of
proposition~\ref{minentrop} to obtain the inequality.
\end{proof}

\section{A  Fokker-Planck model for thermally perfect gases}
\label{sec:FP}

It is difficult to derive a Fokker-Planck model for the distribution
function $f$ with discrete energy levels. We find it easier to
directly derive
a reduced model, by analogy with the reduced BGK
model~(\ref{eq:eq_marginales_fint}) and by
using our previous work~\cite{Mathiaud2017} on
a Fokker-Planck model for polyatomic gases. {We remind that the original Fokker-Planck model for monoatomic gas can be derived from
the Boltzmann collision operator under the assumption of small
velocity changes through collisions and additional equilibrium
assumptions (see~\cite{cercignani}). In practice, the agreement
of this model with the Boltzmann equation is observed even when the
gas is far from equilibrium (see~\cite{Grj2011}, for instance).}

\subsection{A  reduced Fokker-Planck model}

By analogy, now we propose the following reduced Fokker-Planck model for a
diatomic gas with vibrations. Note that now, the model is still with
variables $x$, $v$, and $\eps$: only the discrete energy levels $i$
are eliminated. 
This model is 
\begin{equation}
\begin{aligned}
&  \p_t F + v\cdot\nabla_x F  =  D_F(F,G),\label{eq-FPF}  \\
&  \p_t G + v\cdot\nabla_x G  =  D_G(F,G) ,     
\end{aligned}
\end{equation}
with
\begin{equation}\label{eq-DFDG}   
\begin{split}
& D_F(F,G)=\frac{1}{\tau}\left(\divv\bigl((v-u)F+RT\nabla_v
  F\bigr)\right),\\
& D_G(F,G)= \frac{1}{\tau}\left(\divv\bigl((v-u)G+RT\nabla_v
  G\bigr)\right)
 +\frac2{\tau}\left(e_{int}(T)F-G\right),
\end{split}
\end{equation}
where the macroscopic values are defined as
in~(\ref{eq-macro_reduced}) and ~\eqref{eq-defT}.

\subsection{Properties of the reduced model}

Using direct calculations and dissipation properties as in \cite{Mathiaud2017} we can prove the following proposition.
\begin{proposition} \label{prop:FP}
The collision operator conserves the mass, momentum, and energy:
\begin{equation*}
  \la(1,v)D_F(F,G)\ra_{v} = 0 \quad \text{ and } \quad
\la(\frac12 |v|^2)D_F(F,G) + D_G(F,G)\ra_{v} = 0,
\end{equation*}
the reduced entropy $\mathcal{H}(F,G)$ satisfies the H-theorem:
\begin{equation*}
\partial_t \mathcal{H}(F,G) +\divx\la v H(F,G)\ra_{v}=
  \mathcal{D}(F,G)\leq 0,
\end{equation*}
and we have the equilibrium property
\begin{equation*}
  (D_F(F,G) = 0 \text{ and } D_G(F,G) = 0) \Leftrightarrow
(F = M_{int}[F,G]  \text{ and }  G = e_{int}(T)M_{int}[F,G] ).
\end{equation*}

\end{proposition}
\begin{proof}
The conservation property is the consequence of direct integration
of~\eqref{eq-DFDG}. The equilibrium property can be proved as
follows. 
 To shorten the notations, $M_{int}[F,G]$ will be simply
denoted by $M_{int}$ below, and $e_{int}(T)$ will be simply denoted by
$e_{int}$ as well. Then the collision operators can be written in the compact form
\begin{equation*}
\begin{split}
& D_F(F,G)=\frac{1}{\tau}\nabla_{v} \cdot
\left(  M_{int} \nabla_{v} \frac{F}{M_{int} }  \right),\\
& D_G(F,G)= \frac{1}{\tau}\nabla_{v} \cdot
\left(  M_{int} \nabla_{v} \frac{G}{M_{int} }  \right)
 +\frac2{\tau}\left(e_{int}F-G\right).
\end{split}
\end{equation*}
Then an integration by part gives the following identity for
$D_F(F,G)$: 
\begin{equation*}
  \la D_F(F,G) \frac{F}{M_{int}}\ra_{v} =  
- \frac{1}{\tau}\la \left(\nabla_{v} \frac{F}{M_{int} } \right)^T M_{int} \nabla_{v} \frac{F}{M_{int} } \ra_{v}.
\end{equation*}
Consequently, if $D_F(F,G)=0$, 
since the integrand in the previous relation is a definite positive
form, the gradient is necessarily zero, and hence
$F=M_{int}$. 
For the equilibrium property of $G$, the proof is a bit more
complicated. First, we have 
\begin{equation*}
  \la D_G(F,G) \frac{G}{e_{int}M_{int}}\ra_{v} =  
- \frac{1}{\tau e_{int}}\la \left(\nabla_{v}
\frac{G}{M_{int} }\right)^T  M_{int} \nabla_{v}
\frac{G}{M_{int} } \ra_{v} 
+\la \frac2{\tau}\left(e_{int}F-G\right)\frac{G}{e_{int}M_{int}}\ra_{v}.
\end{equation*}
Consequently, if $D_G(F,G)=0$, and since $F=M_{int}$, we have
\begin{equation*}
\begin{split}
  \frac{1}{e_{int}}\la \left(\nabla_{v}
\frac{G}{M_{int} }\right)^T  M_{int} \nabla_{v}
\frac{G}{M_{int} } \ra_{v}  
& = \frac{2}{\tau}
\la\left(e_{int}M_{int}-G\right)\frac{G}{e_{int}M_{int}}\ra_{v}
\\
&= - \frac{2}{\tau}
\la\left(e_{int}M_{int}-G\right)^2
      \frac{1}{e_{int}M_{int}}
\ra_{v}
+  \frac{2}{\tau}\la  e_{int}M_{int}-G\ra_{v} \\
& \leq \frac{2}{\tau}\la  e_{int}M_{int}-G\ra_{v} 
 = \frac{2}{\tau} (\rho e_{int}-\la G\ra_{v}) = 0,
\end{split}
\end{equation*}
which comes from~\eqref{eq-mtsred} and $F=M_{int}$. Therefore, we
obtain
\begin{equation*}
   \frac{1}{e_{int}}\la \left(\nabla_{v}
\frac{G}{M_{int} }\right)^T M_{int} \nabla_{v}
\frac{G}{M_{int} } \ra_{v}  \leq 0,
\end{equation*}
and again this gives $G = e_{int}M_{int}$,
which concludes the proof of the equilibrium property.

The proof of the H-theorem is much longer. First, by differentiation one gets that the quantity $\mathcal{D}(F,G)=\partial_t \mathcal{H}(F,G) +\divx\la vH(F,G)\ra_{v}$  satisfies:
\begin{eqnarray}
\mathcal{D}(F,G)&=&\la D_1H(F,G)(\partial_tF +v\cdot\nabla_x F )
                      +D_2H(F,G)(\partial_tG +v\cdot\nabla_x G)\ra_{v}\notag\\
&=& \la D_1H(F,G)D_F(F,G)+D_2H(F,G)D_G(F,G)\ra_{v} ,
\end{eqnarray}
from~(\ref{eq:eq_marginales_fint}).
Then the  proof is based on the convexity of $H(F,G)$: while for the BGK we
only used the the first derivatives of $H$, we now use the
positive-definiteness of the Hessian matrix of $H$. To do so we
integrate by parts $\mathcal{D}(F,G)$ and multiply by $\tau$ so that:
\begin{eqnarray*}
\tau\mathcal{D}(F,G)&=&-\sum_{i=1}^3\la\partial_{v_i}(F)D_{11}H(F,G)\left(F(v_i-u_i)+RT\partial_{v_i} F\right)\ra_{v}\\
&&-\sum_{i=1}^3\la\partial_{v_i}(G)D_{21}H(F,G)\left(F(v_i-u_i)+RT\partial_{v_i} F\right)\ra_{v}\\
&&-\sum_{i=1}^3\la\partial_{v_i}(F)D_{12}H(F,G)\left(G(v_i-u_i)+RT\partial_{v_i} G\right)\ra_{v}\\
&&-\sum_{i=1}^3\la\partial_{v_i}(G)D_{22}H(F,G)\left(G(v_i-u_i)+RT\partial_{v_i} G\right)\ra_{v}\\
&&-2\la(e_{int}(T)F-G)\frac 1{RT(G/F)}\ra_{v}
\end{eqnarray*}
To use the positive definiteness of the Hessian matrix $\mathbb{H}$ of
$H$, we introduce the following vector:
\begin{align*}
&   V_i=(F(v_i-u_i)+RT\partial_{v_i} F,G(v_i-u_i)+RT\partial_{v_i} G)    
\end{align*}
and we decompose the partial derivatives of $F$ and $G$ in factor of $D_{11}F$,
$D_{22}F$, $D_{12}F$ as follows:
\begin{align*}
& (\partial_{v_i}(F),\partial_{v_i}(G)) = \frac{1}{RT}V_i
  -(F\frac{v_i-u_i}{RT},G\frac{v_i-u_i}{RT}).
\end{align*}
This gives
\begin{eqnarray*}
\tau\mathcal{D}(F,G)&=&\sum_{i=1}^3\la\left(F\frac{v_i-u_i}{RT}\right)D_{11}H(F,G)\left(F(v_i-u_i)+RT\partial_{v_i} F\right)\ra_{v}\\
&&+\sum_{i=1}^3\la\left(G\frac{v_i-u_i}{RT}\right)D_{21}H(F,G)\left(F(v_i-u_i)+RT\partial_{v_i} F\right)\ra_{v}\\
&&+\sum_{i=1}^3\la\left(F\frac{v_i-u_i}{RT}\right) D_{12}H(F,G)\left(G(v_i-u_i)+RT\partial_{v_i} G\right)\ra_{v}\\
&&+\sum_{i=1}^3\la\left(G\frac{v_i-u_i}{RT}\right)D_{22}H(F,G)\left(G(v_i-u_i)+RT\partial_{v_i} G\right)\ra_{v}\\
&&- \sum_{i=1}^3\la V_i^T\mathbb{H}V_i\ra_{v}\\
&&-2\la(e_{int}(T)F-G)\frac 1{RT(G/F)}\ra_{v}
\end{eqnarray*}

Now this expression can be considerably simplified by using
property~(\ref{eq-FDG}), and we get
\begin{eqnarray*}
\tau\mathcal{D}(F,G)&=&\sum_{i=1}^3\la\left(\frac{v_i-u_i}{RT}\right)\left(F(v_i-u_i)+RT\partial_{v_i} F\right)\ra_{v}\\
&&-\sum_{i=1}^3 V_i^t\mathbb{H}V_i-2\la(e_{int}(T)F-G)\frac 1{RT(G/F)}\ra_{v}.
\end{eqnarray*}
Then the first two terms are simplified by using an integration by
parts and relations~(\ref{eq-mtsred}) and~(\ref{eq-defT}) to get
\begin{eqnarray*}
\tau\mathcal{D}(F,G)&=& \frac{2}{RT}(\rho e_{int}(T) - \la G\ra_{v})-\sum_{i=1}^3 V_i^t\mathbb{H}V_i-2\la(e_{int}(T)F-G)\frac 1{RT(G/F)}\ra_{v}.
\end{eqnarray*}
The terms with the Hessian are clearly negative, since $\mathbb{H}$ is
positive definite. Then we have
\begin{eqnarray*}
\tau\mathcal{D}(F,G)&\leq& \frac{2}{RT}(\rho e_{int}(T) - \la G\ra_{v})-2\la(e_{int}(T)F-G)\frac 1{RT(G/F)}\ra_{v}.
\end{eqnarray*}
Note that from~\eqref{eq-mtsred} the first term can be written as
\begin{equation*}
  \frac{2}{RT}(\rho e_{int}(T) - \la G\ra_{v})  =
  \frac{2}{RT}\la e_{int}(T) F-G\ra_{v}, 
\end{equation*}
and can be factorised with the second term to find
\begin{equation*}
\tau\mathcal{D}(F,G)\leq
2\la(e_{int}(T)F-G)\left( \frac{1}{RT}-\frac 1{RT(G/F)}\right)\ra_{v}.
\end{equation*}
We can now prove that the integrand of the right-hand side is non-positive. Indeed,
assume for instance that the second factor is non-positive, that is to
say $\ds \frac{1}{RT}-\frac 1{RT(G/F)}\leq 0$. Since $e_{int}$ is an increasing function of temperature (see definition~(\ref{eq-etr})), it is now very easy to prove the following relations:
\begin{equation*}
 \frac{1}{RT}-\frac 1{RT(G/F)}\leq 0  \Leftrightarrow   \frac G F\leq e_{int}(T)
\end{equation*}
that is to say the first factor of the integrand is non-negative. 
Consequently, we have proved $\tau\mathcal{D}(F,G)\leq 0$, 
which concludes the proof. 
\end{proof}

\section{Hydrodynamic limits for reduced models}
\label{sec:CE}

With a convenient scaling, the relaxation time $\tau$ of the reduced BGK
model~\eqref{eq:eq_marginales_fint} and
the Fokker-Planck model~(\ref{eq-FPF})) is replaced by $\Kn \tau$, 
where $\Kn$ is the Knudsen number, which can be defined as a ratio
between the mean free path and a macroscopic length scale.
It is then possible to look for macroscopic models derived from BGK and
Fokker-Planck reduced models, in the asymptotic limit of small Knudsen
numbers.  For convenience, these models are re-written below in non-dimensional
form. The BGK model is:
\begin{align}
& \p_t F + v\cdot\nabla_x F  =  \frac{1}{\Kn\tau} \left( M_{int}[F,G] - F \right) \, , \label{eq:BGKf}\\
&  \p_t G + v\cdot\nabla_x G  = \frac{1}{\Kn\tau}  (e_{int}(T) M_{int}[F,G] - G ) \, , \label{eq:BGKg}    
\end{align}
where $M_{int}[F,G]$ can be defined by~(\ref{eq-MvibFG}) with
$R=1$. Similarly, the relations~\eqref{eq-etr}--\eqref{eq-defT}
between the translational, internal, and total energies and the
temperature, have to be read with $R=1$ in  non-dimensional variables. The Fokker-Planck model is
\begin{align}
&  \p_t F + v\cdot\nabla_x F  =  D_F(F,G),\label{eq-FPf}  \\
&  \p_t G + v\cdot\nabla_x G  =  D_G(F,G) \label{eq-FPg},     
\end{align}
with
\begin{equation}\label{eq-DFDGadim}   
\begin{split}
& D_F(F,G)=\frac{1}{\Kn\tau}\left(\divv\bigl((v-u)F+T\nabla_v
  F\bigr)\right),\\
& D_G(F,G)= \frac{1}{\Kn\tau}\left(\divv\bigl((v-u)G+T\nabla_v
  G\bigr)\right)
 +\frac2{\Kn\tau}\left(e_{int}(T)F-G\right).
\end{split}
\end{equation}

\subsection{Euler limit}

In this section, we compute the Euler limit of the two models:
\begin{proposition}
  The mass, momentum, and energy densities $(\rho, \rho u, E = \frac12 \rho
u^2 + \rho e)$ of the solutions of the
  reduced BGK and the Fokker-Planck models satisfy the equations
\begin{equation}\label{eq-euler} 
\begin{split}
& \dt \rho + \nabla_x \cdot \rho u = 0, \\
& \dt \rho u + \nabla_x\cdot  (\rho u\otimes u) + \nabla p = O(\Kn), \\
& \dt E + \nabla_x \cdot (E+p)u =O(\Kn),
\end{split}
\end{equation}
which are the Euler equations, up to $O(\Kn)$. The non-conservative
form of these equations is
\begin{equation} \label{eq-euler_non_cons}
\begin{split}  & \dt \rho + \nabla_x \cdot \rho u = 0, \\
& \rho (\dt u + (u \cdot \nabla_x) u) + \nabla p = O(\Kn), \\
&  \dt T + u\cdot \nabla_x T+\frac{T}{c_v(T)}\divx u =O(\Kn),
\end{split}
\end{equation}
where $c_v(T)=\frac{d}{dT}e(T)$ is the heat capacity at constant volume.
\end{proposition}

\begin{proof}

\

The reduced BGK
  model~\eqref{eq:eq_marginales_fint}
is multiplied by $1$, $v$, and $\frac12 |v|^2$ and integrated
with respect to $v$, which gives the following conservation
laws (with $\sigma(F) = \la F(v-u)\otimes (v-u)  \ra_{v}$  the stress
tensor, and $q(F,G) =  \la  \left(F(\frac12 |v-u|^2)+G\right) (v-u) \ra_{v}
 $ the heat flux):
\begin{equation*} 
\begin{split}
  & \dt \rho + \nabla_x \cdot \rho u = 0, \\
& \dt \rho u + \nabla_x\cdot  (\rho u\otimes u) + \nabla_x \sigma(F) = 0, \\
& \dt E + \nabla_x \cdot Eu + \nabla_x \cdot \sigma(F)   u + \nabla_x \cdot q(F,G)=0.
\end{split}
\end{equation*}

When $\Kn$ is very small, if all the time and space derivatives of $F$
and $G$ are $O(1)$ with respect to $\Kn$,
then~(\ref{eq:BGKf})--(\ref{eq:BGKg})  imply
$F = M_{int}[F,G] + O(\Kn)$ and $G = e_{int}(T)
M_{int}[F,G] +O(\Kn)$ so that $\sigma(F) =
\sigma(M_{int}[F,G]) + O(\Kn) = p I + O(\Kn)$ ,
where $I$ is the unit tensor, and $q(F,G) =
q(M_{int}[F,G],e_{int}(T)M_{int}[F,G])+ O(\Kn) = O(\Kn)$,
which gives the Euler equations~(\ref{eq-euler_non_cons}).
The same analysis can be applied for the reduced Fokker-Planck model~(\ref{eq-FPf})--(\ref{eq-DFDGadim}).
Finally, the non conservative form is readily obtained from the
conservative form. We also get from the non conservative temperature equation: 
\begin{eqnarray}
 \dt e_{int}(T)+ u\cdot \nabla_x
 e_{int}(T)+T\frac{c_v^{int}}{c_v}\divx u =O(\Kn).
\end{eqnarray}
\end{proof}

\subsection{Compressible Navier-Stokes limit}
In this section, we shall prove the following proposition:
\begin{proposition}
  The moments of the solution of the BGK and Fokker-Planck kinetic
  models~(\eqref{eq:eq_marginales_fint})
  and~\eqref{eq-FPF} satisfy,
  up to $O(\Kn^2)$, the Navier-Stokes equations
\begin{equation}\label{eq-ns}
\begin{split}
& \dt \rho + \nabla \cdot \rho u = 0, \\
& \dt \rho u + \nabla\cdot  (\rho u\otimes u) + \nabla p = -\nabla
\cdot \sigma, \\
& \dt E + \nabla \cdot (E+p)u = -\nabla\cdot q - \nabla\cdot(\sigma u),
\end{split}
\end{equation}
where the shear stress tensor and the heat flux are given by
\begin{equation}  \label{eq-fluxes_ns}
\sigma = -\mu \bigl(\nabla u + (\nabla u)^T -\alpha\nabla\cdot
u\bigr), \quad \text{and} \quad  q=-\kappa \nabla \cdot T,
\end{equation}
and  the  values of the viscosity and heat
transfer coefficients (in dimensional variables) are:
\begin{equation}  
\begin{split}
& \mu = \tau p, \quad \text{and} \quad
\kappa = \mu c_p(T) \quad \text{for BGK}, \\
& \mu = \frac{1}{2}\tau p, \quad \text{and} \quad
\kappa = \frac{2}{3}\mu c_p(T) \quad \text{for Fokker-Planck},
\end{split}
\end{equation}
 while the volume viscosity coefficient is $
 \alpha=\frac{c_p(T)}{c_v(T)}-1 $ for both models, and  $c_p(T)=\frac{d}{dT}(e(T)+p/\rho) = c_v(T) +R$ is the heat capacity at constant pressure.
Moreover,
the corresponding Prandtl number is
\begin{equation}  \label{eq-PrM}
\Pr = \frac{\mu c_p(T)}{\kappa}=1 \quad \text{for BGK}, \quad \text{and} \quad \frac{3}{2}
\quad \text{for Fokker-Planck}.
\end{equation}

\end{proposition}

\subsubsection{Proof for the BGK model}
\label{sec:proof-bgk-model}
The usual Chapman-Enskog method is applied as follows. We decompose
$F$ and $G$ as $F = M_{int}[F,G] + \Kn F_1$ and $G = e_{int}(T)
M_{int}[F,G] +\Kn G_1$, which gives
\begin{equation*}
  \sigma(F) = p I + \Kn \sigma(F_1), \qquad \text{and} \qquad
q(F,G) = \Kn q (F_1,G_1).
\end{equation*}
Then we have to approximate $\sigma(F_1)$ and $q (F_1,G_1)$ up to
$O(\Kn)$. This is done by using the previous expansions
and~\eqref{eq:eq_marginales_fint} to
get
\begin{equation*}
\begin{split}
& F_1 = -\tau (\p_t M_{int}[F,G] + v\cdot\nabla_x M_{int}[F,G]) +O(\Kn), \\
& G_1 = -\tau (\p_t e_{int}(T)M_{int}[F,G] + v\cdot\nabla_x e_{int}(T)M_{int}[F,G]) + O(\Kn). \\
\end{split}
\end{equation*}
This gives the following approximations
\begin{equation}
  \sigma(F_1) = -\tau\la
(v-u)\otimes (v-u) (\p_t M_{int}[F,G] + v\cdot\nabla_x M_{int}[F,G])
  \ra_{v}   +O(\Kn),\label{eq-sigmaF1}
\end{equation}  
and
\begin{equation}
\begin{split}
 q(F_1,G_1) =&  -\tau  \la
(v-u)(\frac12 |v-u|^2 )(\p_t M_{int}[F,G] + v\cdot\nabla_x M_{int}[F,G])
 \ra_{v} \\
&  -\tau\la  
(v-u) (\p_t e_{int}(T)M_{int}[F,G] + v\cdot\nabla_x e_{int}(T)M_{int}[F,G] )
\ra_{v}
+ O(\Kn).\label{eq-qF1G1}
\end{split}
\end{equation}

Now it is standard to write $\dt M_{int}[F,G]$ and $\nabla_x M_{int}[F,G]$ as
functions of derivatives of $\rho$, $u$, and $T$, and then to use
Euler equations~\eqref{eq-euler} to write time derivatives as
functions of the space derivatives only. After some algebra, we get
\begin{equation}\label{eq-TMvib} 
 \dt\left(M_{int}(F,G)\right) + v \cdot \nabla_x \left(M_{int}(F,G)\right) = \frac{\rho}{T^{\frac{3}{2}}}M_0(V)\left(A \cdot \frac{\nabla
    T}{\sqrt{T}} + B : \nabla u \right)+ O(\Kn),
\end{equation}
where 
\begin{align*}
& V=\frac{v-u}{\sqrt{T}},  \qquad M_0(V) =
\frac{1}{(2\pi)^{\frac32}}\exp(-\frac{|V|^2}{2})\\
& A = \left(\frac{|V|^2}{2}-\frac{5}{2}\right)V, 
\qquad B = V\otimes V - \left(\frac{1}{c_v}\frac12|V|^2+\frac{e_{int}'(T)}{c_{v}(T)}\right)I.
\end{align*}
Then we introduce~\eqref{eq-TMvib} into~\eqref{eq-sigmaF1} to get
\begin{equation*}
  \sigma_{ij}(F_1) = -\tau \rho T \la
  V_iV_jB_{kl}M_0\ra_{V}\partial_{x_l}u_k + O(\Kn),
\end{equation*}
where we have used the change of variables $v\mapsto V$ in the
integral (the term with $A$ vanishes due to the parity of
$M_0$). Then standard Gaussian integrals (see~appendix~\ref{app:CE}) give
\begin{equation*}
  \sigma(F_1) = -\mu \left(\nabla u + (\nabla u)^T
      -\alpha \nabla \cdot u \, I\right) + O(\Kn), 
\end{equation*}
with $\mu = \tau \rho T$ and $\alpha= \frac{c_p}{c_v}-1$, 
which is the announced result, in a non-dimensional form.

For the heat flux, we use the same technique. First
for $e_{int}(T)M_{int}[F,G]$ we obtain
\begin{equation}\label{eq-TevibMvib} 
 \dt\left(e_{int}M_{int}(F,G)\right) + v \cdot \nabla_x \left(e_{int}M_{int}(F,G)\right) = \frac{\rho}{T^{\frac{3}{2}}}M_0(V)\left(\tilde A \cdot \frac{\nabla
    T}{\sqrt{T}} + \tilde B : \nabla u\right) + O(\Kn),
\end{equation}
where
\begin{align*}
&   \tilde A = \left(\frac{|V|^2}{2}-\frac{5}{2}+\frac{Te_{int}'(T)}{e_{int}}\right)V ,\\
&  \tilde B = V\otimes V - \left(\frac{1}{c_v}\frac12|V|^2+\frac{e_{int}'(T)}{c_{v}(T)}+\frac{Te_{int}'(T)}{c_v(T)e_{int}}\right)I.
\end{align*}
Then $q(F_1,G_1)$ as
given in~\eqref{eq-qF1G1} can be reduced to
\begin{equation*}
\begin{split}
q_i(F_1,G_1) & = 
-\tau \rho T\left( 
\la  \frac{1}{2}|V|^2V_iA_jM_0\ra_{V}
+ \la  V_i J A_jM_0\ra_{V}
\right)\partial_{x_j}T \\
& \qquad 
- \tau \rho 
 \la  V_i\tilde{A}_jM_0\ra_{V}\partial_{x_j}T.
\end{split}
\end{equation*}
Using again Gaussian integrals , we get
\begin{equation*}
q(F_1,G_1)=-\kappa\nabla_x T,  
\end{equation*}
where $\kappa = \mu c_p(T)$ with  $c_p(T)= \frac{d}{dT}(e(T) + \frac{p}{\rho}) = \frac{5}{2} +
e_{int}'(T)= 1 + c_v(T)$ in a non-dimensional form.

\subsubsection{Proof for the Fokker-Planck model}

Here, we rather use the decomposition $F = M_{int}(1 + \Kn F_1)$ and $G = e_{int}
M_{int} (1+\Kn G_1)$, which gives 
\begin{equation*}
  \sigma(F) = p I + \Kn \sigma(M_{int}F_1) \quad \text{and} \quad
  q(F,G) = \Kn q (M_{int}F_1,e_{int}M_{int}G_1),
\end{equation*}
in which, for clarity, the dependence of $M_{int}$ on $F$ and $G$ has been omitted,
and the dependence of $e_{int}$ on $T$ as well.
Finding $F_1$ and $G_1$ is less simple than for the BGK
model: however, the computations are very close to what is done in the
standard monatomic Fokker-Planck model (see~\cite{Mathiaud2016} for instance), so
that we only give the main steps here (see appendix~\ref{app:CE} for
details).

First, the decomposition is injected into~(\ref{eq-DFDGadim}) to get
\begin{align*}
&  D_F(F,G) = \frac{1}{\tau}M_{int}L_F(F_1) + O(\Kn) , \\
&  D_G(F,G) = \frac{1}{\tau}e_{int}M_{int}L_G(F_1,G_1) + O(\Kn) , 
\end{align*}
where $L_F$ and $L_G$ are linear operators defined by
\begin{equation} \label{eq-LFLG} 
\begin{split}
& L_F(F_1) = \frac{1}{M_{int}}\Bigl(\divv (TM_{int}\nabla_v F_1)\Bigr),\\
& L_G(F_1,G_1) = \frac{1}{e_{int}M_{int}}
 \Bigl( \divv (Te_{int}M_{int}\nabla_v G_1)
+2(F_1-G_1)
 \Bigr).
\end{split}
\end{equation}

Then the Fokker-Planck equations~(\ref{eq-FPf})-(\ref{eq-FPg}) suggest to look
for an approximation of $F_1$ and $G_1$ up to $O(\Kn)$ as solutions of 
\begin{align*}
&  \p_t M_{int} + v\cdot\nabla_x M_{int} =
  \frac{1}{\tau}M_{int}(F,G)L_F(F_1) \\
& \p_t e_{int}M_{int} + v\cdot\nabla_x e_{int}M_{int}
 =   \frac{1}{\tau}e_{int}M_{int}(F,G)L_G(F_1,G_1).
\end{align*}
By using~(\ref{eq-TMvib})-(\ref{eq-TevibMvib}), these relations are
equivalent, up to
another $O(\Kn)$ approximation, to 
\begin{equation}\label{eq-F1G1FP} 
L_F(F_1) = \tau \left(A \cdot \frac{\nabla T}{\sqrt{T}} 
+ B:\nabla u\right), 
\quad \text{ and } \quad
 L_G(F_1,G_1) = \tau \left(\tilde{A} \cdot \frac{\nabla T}{\sqrt{T}} 
+ \tilde{B}:\nabla u\right),
\end{equation}
where $A$, $B$, $\tilde{A}$, and $\tilde{B}$ are the same as for the
BGK equation in the previous section.

Now, we rewrite $L_F(F_1)$ and $L_G(F_1,G_1)$, defined in~\eqref{eq-LFLG}, by using the change of
variables $V = \frac{v-u}{\sqrt{T}}$ to get
\begin{equation*}
\begin{split}
& L_F(F_1) = -V\cdot \nabla_V F_1 + \nabla_V \cdot (\nabla_V F_1) 
 , \\
& L_G(F_1,G_1) = L_F(G_1) + 2(F_1-G_1).
\end{split}
\end{equation*}
Then simple calculation of derivatives show that $A$, $B$, $\tilde{A}$, and $\tilde{B}$ satisfy the
following properties
\begin{align*}
&   L_F(A) = -3 A, \qquad L_F(B) = -2 B, \\
&   L_G(A,\tilde{A}) = -3 \tilde{A}, \qquad L_G(B,\tilde{B}) = -2 \tilde{B}. 
\end{align*}
Therefore, we look for $F_1$ and $G_1$ as solution
of~(\ref{eq-F1G1FP}) under the following form
\begin{equation*}
  F_1 = a A \cdot \frac{\nabla T}{\sqrt{T}} + b B:\nabla u \quad
  \text{ and } \quad 
   G_1 = \tilde{a} \tilde{A} \cdot \frac{\nabla T}{\sqrt{T}} + \tilde{b} \tilde{B}:\nabla u ,
\end{equation*}
and we find $\tilde{a}=a=-1/3$ and $\tilde{b}=b=1/2$.

Finally, using these relations into $\sigma$ and $q$ and using some
Gaussian integrals (see appendix~\ref{app:CE}) give
\begin{equation*}
  \sigma(M_{int}F_1) = -\mu \left(\nabla u + (\nabla u)^T
      -\alpha \nabla \cdot u \, I\right) \quad \text{ and } \quad 
   q (M_{int}F_1,e_{int}M_{int}G_1) = -\kappa\nabla_x T,
\end{equation*}
where $\alpha= \frac{c_p}{c_v}-1$, $\mu = \frac{\tau}{2}\rho T$, and
$\kappa = \frac{2}{3}\mu c_p(T)$, which is the announced result, in a
non-dimensional form.

\section{Extension of  the model}
\label{sec:extens}

\subsection{Extension to several type of energies}

The model we present in this paper recovers Navier-Stokes with  potentially false Prandtl number for both BGK and Fokker-Planck models as usual. If one wants to describe more precisely relaxation phenomena of molecules, one has to consider each independent internal energies which means for instance that rotational energy and vibrational energies have to be separated to capture them correctly. Before going to ESBGK or ES Fokker-Planck models we now present the framework that should allow to go further. Let us define $e_{int}^a..,e_{int}^n$, $n$ independent terms of the internal energy (rotation energy, vibrations energy, electronic energy..) depending
on temperature through strictly increasing functions. We can define   entropy's $s_{int}^a..,s_{int}^n$ associated to each energy  satisfying $T^i_{int}ds_{int}^i=e_{int}^i$. As before we  define $F$ the function transporting the velocities as well as $G^a,...,G^n$ the functions that transport other energies. The macroscopic variables 
are now obtained through $F$
and $G^a,...,G^n$ only, as it is shown in the following proposition.
\begin{proposition}[Moments of the reduced distributions]
The macroscopic variables $\rho$, $u$, and $e$  are defined through
\begin{equation}  \label{eq-mtsred2} 
 \rho = \la F \ra_{v},\qquad
 \rho u = \la v F \ra_{v},\qquad
 \rho e = \la (\frac12(v-u)^2)  F\ra_{v}+\la (G^a+..+G^n)\ra_{v}.
\end{equation}
\end{proposition}

Now it is possible to write a reduced entropy as a function of $F$
and $G^a,...,G^n$  only, as it is shown in the following proposition.
\begin{proposition}[Entropy]\label{minentropgen}
We define the following reduced entropy for $F$,$G^a$...$G^n$ $\mathcal{H}(F,G)$: 
\begin{eqnarray}
\mathcal{H}(F,G^a..,G^n)= \la F\log(F)- F\frac{s_{int}^a}R\left(\frac {G^a} F\right)-...-\- F\frac{s_{int}^n}R\left(\frac {G^n} F\right) \ra_{v}. \label{eq-HFGn}
\end{eqnarray}
\begin{enumerate}
\item  The partial derivatives of $H$ computed at $(F,G)$ are:
\begin{eqnarray}  
  &&D_F H(F,G)=1+\log(F)+\sum_{i=a}^n\left(\frac{G^a}{RT_{int}^a(G^a/F)F}-\frac{s_{int}^a}R\left(\frac {G^a} F\right)\right),\\
&& D_{G^a} H(F,G)=-\frac1{RT_{int}^a(G^a/F)}.
\end{eqnarray}

\item We note $\ds \mathbb{H}$
 the Hessian matrix of $H$ which can be computed through:
\begin{eqnarray}
 D_{F,F}H &=&\frac 1 F +\sum_{i=a}^n\frac{{G^i}^2}{F^3{c_v^i}_{int}R{T^i_{int}}^2(G^i/F)}, \\
D_{G^a,F}= D_{F,G^a}H&=&-\frac {G^a}{F^2{c_v^a}_{int}R{T^a_{int}}^2(G^a/F)}\\
 D_{G^a,G^b}H&=&0 (a\neq b), \\
 D_{G^a,G^a}H&=&\frac{1}{F{c_v^a}_{int}R{T^a_{int}}^2(G^a/F)}
\end{eqnarray}
 Moreover, we have the following equality's:
\begin{equation}\label{eq-FDGn} 
 FD_{F,F}H+\sum_{i=a}^n G^iD_{F,G^i}H = 1 \hspace{0.2cm}, \hspace{0.2cm} FD_{F,G^a}H+G^aD_{G^a,G^a}H = 0.
\end{equation}

\item The function $(F,G^a,...,G^n) \mapsto H(F,G^a,...,G^n)$ is convex.

\item  Let $(F,G^a,...,G^n)$ be  reduced distributions and $\rho$, $\rho u$, and $\rho e$ their moments as defined by~\eqref{eq-mtsred2}. Let ${\cal S}$ be the convex set defined by
\begin{equation*}
    \mathcal{S}=\left\{(F_1,G_1^a..,G_1^n) \text{ such that} \la F_1\ra_{v}=\rho ,\quad
  \la vF_1\ra_{v}=\rho u, \quad 
  \la \frac12|v|^2 F_1+G_1^a+..+G_1^n\ra_{v} = \rho e \right\}.
\end{equation*} The minimum of 
  $\mathcal{H}$ on ${\cal S}$
is obtained for  $(M_{int}(F,G), e^a_{int}(T)M_{int}(F,G),...,e^n_{int}(T)M_{int}(F,G))$
with
\begin{equation}  \label{eq-MintFGn}
M_{int}[F,G]=\frac{\rho}{\sqrt{2\pi RT}^3} \exp\left(-\frac{|v - u|^2}{2 RT}\right)  
 \end{equation}
where $e^a_{int}(T)$ is the equilibrium internal $a$ energy obtained for $T$.

\item For every $(F_1,G_1^a,...,G_1^n)$ in $\mathcal{S}$, we have
\begin{eqnarray}
 0 & \ge&  \ds H(M_{int}(F,G^a,...,G^n), e_{int}(T)M_{int}(F,G^a,...,G^n)) -H(F_1,G_1^a,...,G_1^n) \notag \\
&\ge&  D_1H(F,G)(M_{int}(F,G^a,...,G^n)-F) \notag\\
&&+\sum_{i=a}^nD_{G^i}(F,G^a,...,G^n)(e^i_{int}(T)M_{int}(F,G^a,...,G^n))-G^i)
\end{eqnarray}
\end{enumerate}

\end{proposition}
\begin{proof}
The proof is the same as the one with one energy. The only tricky part (that we prove here) is
that the Hessian is positive definite. The quadratic form associated to $\mathbb H$ is clearly positive definite on vectors of the form $(0,x_1...,x_n)$ because the diagonal terms of the matrix are strictly positive on this subspace so the Hessian have at least $n$ strictly positive eigenvalues.  To ensure that the Hessian matrix is positive definite it is sufficient to have a strictly positive determinant. Developing the determinant one gets:
\begin{eqnarray*}
\det(\mathbb H)&=&D_{F,F}(H)\prod_{i=a}^n D_{G^a,G^a}(H)-\sum_{i=a}^n D_{G^i,F}^2(H)\prod_{i\neq j} D_{G^j,G^j}(H)\\
&=&\frac1F \prod_{i=a}^n D_{G^a,G^a}(H) -\sum_{i=a}^n \frac{G^i}{F}(H)D_{G^i,F}(H) \prod_{i=a}^n D_{G^a,G^a}(H) \\
&&- \sum_{i=a}^n \left(-\frac{G^i}{F}D_{G^i,G^i}(H)D_{G^i,F}(H)\right)\prod_{i\neq j} D_{G^j,G^j}(H)\\
&=&\frac1F \prod_{i=a}^n D_{G^a,G^a}(H) -\sum_{i=a}^n \frac{G^i}{F}D_{G^i,F}(H) \prod_{i=a}^n D_{G^a,G^a}(H) + \left(\sum_{i=a}^n \frac{G^i}{F}D_{G^i,F}(H)\right)\prod_{i=a}^n D_{G^i,G^i}(H)\\
&=&\frac1F \prod_{i=a}^n D_{G^a,G^a}(H)\\
&>&0
\end{eqnarray*}
so that the Hessian is positive definite. Equilibrium property is the same as with one energy and convex properties are obtained thanks to the Hessian
\end{proof}

Thanks to this framework we are able able to give the $BGK$ model and the Fokker-Planck model associated to $n$ energies as well as their Chapman-Enskog expansion. We do not give the proof since there are exactly the same as before. 

\subsection{BGK model and its hydrodynamic limit for $n$ energies}

For physics considerations, it is interesting to reduce  complex
kinetic models by using the usual reduced distribution
technique~\cite{HH_1968}. Even for perfect gases, some degenerate models of energy cannot be described easily through equilibrium (extensions of Maxwellians is not very clear when ones deals with partial degrees of freedom of internal energy. In this paper we propose to use two reduced functions $F$ and $G$ that will transport energy. $F$ is transporting the translational energy whereas $G$ transport the reminder of internal energy. More precisely they are defined through:
\begin{align}
& \p_t F + v\cdot\nabla_x F  =  \frac1 \tau \left( M_{int}[F,G] - F \right) \, , \label{eq:eq_marginales_fintn}\\
&  \p_t G^i + v\cdot\nabla_x G^i  = \frac1 \tau ( e_{int}^a M_{int}[F,G^a,..,G^n] - G^i ) \hspace{1cm}  \forall a\leq i\leq n \label{eq:eq_marginales_gintn}    
\end{align}

where the reduced Maxwellian is
\begin{equation*}
M_{int}[F,G]=\frac{\rho}{\sqrt{2\pi RT}^3} \exp\left(-\frac{|v - u|^2}{2 RT}\right),
\end{equation*}

and the macroscopic quantities are defined by

\begin{equation}\label{eq-macro_reducedn} 
 \rho = \la F \ra_{v}, \qquad \rho u= \la vF \ra_{v},
 \qquad 
 \rho e= \la (\frac12(v-u)^2)F \ra_{v}+\sum_{i=a}^n\la G_i\ra_{v},
\end{equation}

and $T$ is still defined by~\eqref{eq-defT}.

System~(\ref{eq:eq_marginales_fintn}--\ref{eq:eq_marginales_gintn})
naturally satisfies local conservation laws of mass, momentum, and
energy. Moreover, the H-theorem holds with the reduced entropy 
$H(F,G)$ as defined in~(\ref{eq-HFG}). Indeed, we recover the two following propositions:
\begin{proposition}
The reduced BGK
system~(\ref{eq:eq_marginales_fintn}--\ref{eq:eq_marginales_gintn})
satisfies the H-theorem
\begin{equation*}
\partial_t {\cal H}(F,G^a,...,G^n) +\divx\la v H(F,G^a,...,G^n)\ra_{v}\leq 0,
\end{equation*}
where  ${\cal H}(F,G^a,...,G^n)$ is the reduced entropy defined in~\eqref{eq-HFGn}.
\end{proposition}

\begin{proposition}
  The moments of the solution of the BGK 
  models~\eqref{eq:eq_marginales_fintn}-\eqref{eq:eq_marginales_gintn}
  satisfy,
  up to $O(\Kn^2)$, the Navier-Stokes equations
\begin{equation}
\begin{split}
& \dt \rho + \nabla \cdot \rho u = 0, \\
& \dt \rho u + \nabla\cdot  (\rho u\otimes u) + \nabla p = -\nabla
\cdot \sigma, \\
& \dt E + \nabla \cdot (E+p)u = -\nabla\cdot q - \nabla\cdot(\sigma u),
\end{split}
\end{equation}
where the shear stress tensor and the heat flux are given by
\begin{equation}  
\sigma = -\mu \bigl(\nabla u + (\nabla u)^T -\alpha\nabla\cdot
u\bigr), \quad \text{and} \quad  q=-\kappa \nabla \cdot T,
\end{equation}
and where the following values of the viscosity and heat
transfer coefficients (in dimensional variables) are
\begin{equation}  
\begin{split}
& \mu = \tau p, \quad \text{and} \quad
\kappa = \mu c_p(T),
\end{split}
\end{equation}
 while the volum viscosity coefficient is $
 \alpha=\frac{c_p(T)}{c_v(T)}-1 $  and  $c_p(T)=\frac{d}{dT}(e(T)+p/\rho) = c_v(T) +R$ is the heat capacity at constant pressure.
Moreover,
the corresponding Prandtl number is
\begin{equation}  
\Pr = \frac{\mu c_p(T)}{\kappa}=1
\end{equation}

\end{proposition}

\subsection{Fokker-Planck model and its hydrodynamic limit for $n$ energies}

By analogy,  we propose the following reduced Fokker-Planck model:
\begin{align}
&  \p_t F + v\cdot\nabla_x F  =  D_F(F,G^a,...,G^n),\label{eq-FPFn}  \\
&  \p_t G^i + v\cdot\nabla_x G^i =  D_{G^i}(F,G^i,...,G^n) \label{eq-FPGn} \hspace{1cm}\forall a\leq i\leq n ,     
\end{align}
with
\begin{equation}\label{eq-DFDGn}   
\begin{split}
& D_F(F,G^a,...,G^n)=\frac{1}{\tau}\left(\divv\bigl((v-u)F+RT\nabla_v
  F\bigr)\right),\\
& D_{G^i}(F,G^a,...,G^n)= \frac{1}{\tau}\left(\divv\bigl((v-u)G^i+RT\nabla_v
  G^i\bigr)\right)
 +\frac2{\tau}\left(e^i_{int}(T)F-G^i\right),
\end{split}
\end{equation}
where the macroscopic values are defined as
in~(\ref{eq-macro_reducedn}) and ~\eqref{eq-defT}.
Using direct calculations and dissipation properties  we can prove the following propositions.
\begin{proposition} \label{prop:FPn}
The collision operator conserves the mass, momentum, and energy:
\begin{equation*}
  \la(1,v)D_F(F,G^a,...,G^n)\ra_{v} = 0 \quad \text{ and } \quad
\la\frac12 |v|^2D_F((F,G^a,...,G^n) + D_G(F,G^a,...,G^n)\ra_{v} = 0,
\end{equation*}
the reduced entropy $\mathcal{H}(F,G^a,...,G^n)$ satisfies the H-theorem:
\begin{equation*}
\partial_t \mathcal{H}(F,G^a,...,G^n) +\divx\la v H(F,G^a,...,G^n)\ra_{v}\leq 0,
\end{equation*}
and we have the equilibrium property
\begin{eqnarray*}
 && (D_F(F,G^a,...,G^n) = 0 \text{ and } \forall i, D_{G^i}(F,G^a,...,G^n) = 0) \\&\Leftrightarrow&
(F = M_{int}(F,G^a,...,G^n) \text{ and }  \forall i, G^i = e_{int}^i(T)M_{int}(F,G^a,...,G^n) ).
\end{eqnarray*}

\end{proposition}

\begin{proposition}
  The moments of the solution of the Fokker-Planck kinetic
  model \eqref{eq-FPFn}-\eqref{eq-FPGn} satisfy,
  up to $O(\Kn^2)$, the Navier-Stokes equations
\begin{equation}
\begin{split}
& \dt \rho + \nabla \cdot \rho u = 0, \\
& \dt \rho u + \nabla\cdot  (\rho u\otimes u) + \nabla p = -\nabla
\cdot \sigma, \\
& \dt E + \nabla \cdot (E+p)u = -\nabla\cdot q - \nabla\cdot(\sigma u),
\end{split}
\end{equation}
where the shear stress tensor and the heat flux are given by
\begin{equation}  
\sigma = -\mu \bigl(\nabla u + (\nabla u)^T -\alpha\nabla\cdot
u\bigr), \quad \text{and} \quad  q=-\kappa \nabla \cdot T,
\end{equation}
and where the following values of the viscosity and heat
transfer coefficients (in dimensional variables) are
\begin{equation}  
\begin{split}
& \mu = \frac{1}{2}\tau p, \quad \text{and} \quad
\kappa = \frac{2}{3}\mu c_p(T) ,
\end{split}
\end{equation}
 while the volumic viscosity coefficient is $
 \alpha=\frac{c_p(T)}{c_v(T)}-1 $ for both models, and  $c_p(T)=\frac{d}{dT}(e(T)+p/\rho) = c_v(T) +R$ is the heat capacity at constant pressure.
Moreover,
the corresponding Prandtl number is
\begin{equation} 
\Pr = \frac{\mu c_p(T)}{\kappa}=\frac{3}{2}
\end{equation}

\end{proposition}

\subsection{Comments and application to the vibrational case}

In the previous subsection we have explained how we can try to capture every kind of energy as long as they are strictly increasing functions of temperatures. We also have constructed an entropy adapted to this situation but to fully use the result one will have to create ESBGK or ES-Fokker Planck like models to capture different relaxations times. We now explain how to use this extension  for a diatomic vibrational gas. Such a gas owns a translational, a rotational and a vibrational energy defined as functions of temperatures through:
\begin{align}
&  e_{tr}(T) =\frac32  RT \label{eq-etr1},\quad
  e_{rot}(T) =  RT ,\quad
  e_{vib}(T) =\frac{RT_0}{e^{T_0/T}-1},   
\end{align}
The associated macroscopic entropy's for internal degrees of freedom are 
\begin{align}
&  s_{rot}(e) = R\ln(e), \label{eq-srot1} \quad 
  s_{vib}(e) =\left(\frac{e}{T_0}+R\right)\ln \left(\frac{e+RT_0}{RT_0}\right)-\frac{e}{T_0}\ln\left(\frac{e}{RT_0}\right), 
\end{align}
which leads to the following kinetic entropy:
\begin{eqnarray*}
&&\mathcal{H}(F,G^{rot},G^{vib})\\
&=& \la F\log(F)- F\frac{s^{rot}}R\left(\frac {G^{rot}}F\right) 
-F\frac{s^{vib}}R\left(\frac {G^{vib}} F\right) \ra_{v} \\
&=& \la F\log(F)+ F\ln\left(\frac F{G^{rot}}\right) +F\ln\left(\frac {RT_0F}{RT_0F+G^{vib}}\right)+\frac {G^{vib}}{RT_0}\ln\left(\frac{G^{vib}}{RT_0F+G^{vib}}\right) \ra_{v}.
\end{eqnarray*}
The expression of the vibration's entropy recovers the one given in \cite{hdr,mathiaud2019BGK} and the expression for rotations  the one in \cite{esbgk_poly}.

\section{Conclusion and perspectives}
\label{sec:conclusion}

In this paper, we have proposed to different models (BGK and
Fokker-Planck) of the Boltzmann equation that allow for thermally perfect gases. These models satisfy the conservation and
entropy property (H-theorem) and are using reduced distribution functions with only velocity as a kinetic variable. The low complexity of the reduced BGK model can make it attractive to be implemented
in a deterministic code, while the Fokker-Planck model can be easily
simulated with a stochastic method.
Of course, since these models are based on a single time relaxation,
they cannot allow for multiple relaxation times scales but we have made ground for standard procedures  like the ellipsoidal-statistical approach, already used to
correct the Prandtl number of the BGK model~\cite{esbgk_poly} and
Fokker-Plank models~\cite{hdr} by already defining models with one equation for each kind of energy  in the last section of this  paper.

\appendix
\section{Gaussian integrals and other summation formula}
\label{app:CE}

In this section, we give some integrals and summation formula that are
used in the paper. 

First, we remind the definition of the absolute Maxwellian $M_0(V) =
\frac{1}{(2\pi)^{\frac32}}\exp(-\frac{|V|^2}{2})$. We denote by
$\cint{\phi} = \int_{\R^3}\phi(V)\, dV$ for any function $\phi$. It is standard to
derive the following integral relations (see~\cite{chapmancowling},
for instance and note that some computations are redundant), written with the Einstein notation:
\begin{align*}
&   \cint{M_0}_V = 1, \\
&   \cint{V_iV_jM_0}_V = \delta_{ij}, \qquad \cint{V_i^2M_0}_V = 1,
  \qquad \cint{|V|^2M_0}_V = 3, \\
& \cint{V_i^2V_j^2M_0}_V = 1 + 2\, \delta_{ij},  \qquad \cint{V_iV_jV_kV_lM_0}_V = \delta_{ij}\delta_{kl}  +
  \delta_{ik}\delta_{jl}  + \delta_{il}\delta_{jk}  \\
& \cint{V_iV_j|V|^2M_0}_V = 5 \,\delta_{ij},  \qquad \cint{|V|^4M_0}_V
  = 15, \\
& \cint{V_iV_j|V|^4M_0}_V = 35 \,\delta_{ij},  \qquad \cint{|V|^6M_0} = 105,
\end{align*}
while all the integrals of odd power of $V$ are zero.
From the previous Gaussian integrals, it can be shown that for any
$3\times 3$ matrix $C$, we have
\begin{equation*}
\cint{V_iV_jC_{kl}V_kV_lM_0}_V = C_{ij} + C_{ji} + C_{ii}\delta_{ij}.
\end{equation*}

 \bibliographystyle{unsrt}
 \bibliography{biblio}

\end{document}